\documentclass[journal]{IEEEtran}
\usepackage{amssymb,amsthm,amsmath,amsfonts}
\usepackage{algorithmic}
\usepackage{algorithm}
\usepackage{array}
\usepackage[caption=false,font=normalsize,labelfont=sf,textfont=sf]{subfig}
\usepackage{textcomp}
\usepackage{stfloats}
\usepackage{url}
\usepackage{verbatim}
\usepackage{graphicx}
\usepackage{cite}
\usepackage{color}
\newtheorem{theorem}{Theorem}
\newtheorem{lemma}{Lemma}
\newtheorem{remark}{Remark}

\begin{document}

\title{Robust Recovery of Sparse Support in Constrained Group Testing}

\author{Jianing Li, Li Chai,~\IEEEmembership{Senior Member,~IEEE}, Xinyao Rao, Hailin Zhang
\thanks{This work was supported by National Science Foundation of China under grant 62550085.}
\thanks{Jianing Li, Li Chai, Xinyao Rao, Hailin Zhang are with the College of Control Science and Engineering, Zhejiang University, Hangzhou 310027, China (Email:lijn202409@zju.edu.cn; chaili@zju.edu.cn; rxyyy@zju.edu.cn; hl\_zhang@zju.edu.cn).}}



\maketitle

\begin{abstract}
In the early stage of a pandemic, rapidly identifying a small number of infected individuals through large-scale screening is critical for pandemic control. Group testing has been widely used to improve testing efficiency and numerous studies have investigated the problem under noisy measurements, typically modeled as bit-flipping of test outcomes. However, these methods do not consider the constraints imposed by dilution, pool size, and the limit of detection ($\mathrm{\mathbf{LOD}}$), which can lead to false negatives when the viral load in a pool falls below the $\mathrm{\mathbf{LOD}}$. In addition, liquid dispensing errors, common in laboratory settings, affects diluted viral loads in a nonlinear manner. In this work, we introduce a novel measurement model that characterizes the process of sample pooling and dilution, incorporating $\mathrm{\mathbf{LOD}}$-induced binary quantization as well as liquid dispensing errors. For the case with known sparsity level, we propose a low-complexity decoding algorithm and provide theoretical guarantees for exact support recovery under both noiseless and noisy settings. For the case with unknown sparsity level, we develop a blind support recovery algorithm, along with a heuristic variant to enhance robustness, which can achieve exact support recovery with only $\mathcal{O}\left(k\log n\right)$ measurements. Extensive simulations show that the proposed algorithms outperform existing combinatorial group testing algorithms, validating the effectiveness, efficiency and robustness in large-scale screening.
\end{abstract}


\section{Introduction}
Over the past two decades, recurrent outbreaks of respiratory diseases, including SARS, MERS, and COVID-19, have posed severe challenges to global public health systems \cite{zhou2020pneumonia}. It is widely recognized that large-scale screening plays a critical role in containing the spread of such diseases \cite{sanchez2020covid,atkeson2024economic}, as timely identification of infected individuals enables effective isolation and intervention. However, practical constraints, such as limited testing capacity, reagent shortages, and financial costs, often make massive individual testing infeasible \cite{cheng2020diagnostic,cleary2021using}. These challenges become more pronounced in the early stage of an outbreak, when the disease prevalence is extremely low. In such scenarios, with only a small number of infections among a large population ($k\ll n$), there is an ongoing demand for efficient and scalable testing strategies that can rapidly identify all positive cases using as few tests as possible.

Group testing, first introduced by Dorfman \cite{dorfman1943detection} in 1943, has emerged as an effective approach to improve testing efficiency of large-scale screening. In Dorfman-style (adaptive) group testing methods \cite{dorfman1943detection,hwang1972method,du1999combinatorial,westreich2008optimizing,smith2009use,bilder2020tests,da2022simulation}, specimens from multiple individuals are combined into pools, and only a small number of pools are tested. If a pool tests negative, all samples within that pool are declared uninfected, whereas a positive outcome indicates the presence of at least one infected sample. In that case, all samples in the positive pools are subsequently tested individually. While this simple yet powerful strategy enables a large population to be screened using a limited number of assays, it requires multiple rounds of testing, leading to increased time delays and operational complexity \cite{da2022simulation}. These limitations have motivated the development of non-adaptive group testing schemes \cite{chan2011non,aldridge2014group,chan2014non,cai2017efficient,8374880,gandikota2019nearly,8771121,scarlett2020noisy,price2023fast,atia2012boolean,schumacher2020statistics}, in which all pools are designed and tested in a single stage. Each pool is formed by combining samples according to a pooling matrix $A\in\left\{0,1\right\}^{m\times n}$, where $m$ denotes the number of tests. The resulting test outcomes $y\in\left\{0,1\right\}^m$ satisfy
\begin{equation}
  y_i=\bigvee_{j=1}^{n} A_{i,j} v_j, \quad i=1,\ldots,m,
\end{equation}
where $\bigvee$ is the Boolean OR operation, $v\in\left\{0,1\right\}^{n}$ represents the sample status vector, with $v_j=1$ indicating that the $j$-th sample is positive and $v_j=0$ otherwise. The identification of positive samples is then performed via a decoding algorithm based on the binary test outcomes $y$ and the pooling matrix $A$. Among the most representative approaches \cite{chan2011non,aldridge2014group,chan2014non}, classical combinatorial group testing algorithms such as COMP, DD and SCOMP are designed for the noiseless setting. The Combinatorial Orthogonal Matching Pursuit (COMP) algorithm is the simplest, eliminating samples appearing in negative pools and declaring the remaining as possible positives. The Definite Defective (DD) algorithm improves upon COMP by declaring a sample as positive if it is the sole member of a positive pool, thereby reducing false positives. Sequential COMP (SCOMP) further refines the results of DD by iteratively selecting samples that best explain the observed positive tests. These combinatorial group testing algorithms primarily rely on Boolean operations and achieve low computational complexity. However, they are sensitive to noise and may suffer degraded recovery performance in practical settings.

Accordingly, several methods \cite{chan2011non,chan2014non,cai2017efficient,8374880,gandikota2019nearly,8771121,scarlett2020noisy,price2023fast} have been proposed for the noisy setting, where the test outcomes may flip with a certain probability. Early extensions adapt the above combinatorial algorithms by relaxing their strict thresholding rules. Chan et al. \cite{chan2011non, chan2014non} developed several variants (e.g., noisy COMP) to address noisy measurements. The noisy DD algorithm was studied in \cite{scarlett2020noisy} under Z-channel and reverse Z-channel noise models. More advanced frameworks \cite{cai2017efficient,8771121,price2023fast} further jointly design the pooling matrix and decoding algorithm. Cai et al. \cite{cai2017efficient} propose a decoding framework based on GROTESQUE tests, which combine multiplicity testing and expander-code-based localization to iteratively identify positives, assuming the sparsity level $k$ is known a priori. Lee et al. \cite{8771121} proposed SAFFRON, which uses sparse-graph codes to construct the pooling matrix and recovers the support set via an iterative peeling process, achieving a computational complexity of $\mathcal{O}\left(k\log n\right)$. This computational efficiency is obtained at the cost of increased sample complexity of $\mathcal{O}\left(k\log k\log n\right)$ for exact support recovery with high probability. Price et al. \cite{price2023fast} proposed a fast splitting framework for sparsity-constrained and noisy group testing. By recursively partitioning the item set and eliminating negative subsets, the algorithm achieves near-optimal test complexity with a decoding time that matches the number of tests. In addition to the above noise models, there are other interesting models as discussed in \cite{aldridge2019group}. A prominent example is the dilution noise model proposed in \cite{atia2012boolean}, where each positive sample in a pool is independently ``diluted'' (and hence behave as if it were negative) with probability $q$. Consequently, for a pool containing $u$ positives, the probability that all positives are diluted and the pool tests negative is $q^u$. Similarly, Schumacher et al. \cite{schumacher2020statistics} analyzed the noisy COMP algorithm under a noise model where the probability that a pool containing positives tests negative depends on the number of positives in the pool.

Recently, several studies proposed compressed sensing (CS)-based methods, which exploits signal sparsity to estimate viral loads using only $\mathcal{O}\left(k\log n\right)$ measurements \cite{donoho2006compressed,candes2008introduction}. They take advantage of quantitative measurements \cite{ghosh2021compressed,cohen2021multi,bharadwaja2022recovery,das2024performance}, such as viral loads of pools or cycle threshold (CT) values, which can be modeled as a linear system:
\begin{equation}
  z=Ax+\eta,
\end{equation}
where $z\in\mathbb{R}^m$ denotes the quantitative measurements, $x\in\mathbb{R}^n$ represents the viral loads of $n$ samples with ${\left\| x \right\|_0 = k}$, and $\eta\in\mathbb{R}^n$ is the noise vector. Sparse signal recovery is then performed via convex optimization (e.g., Basis Pursuit De-Noising, LASSO), enabling exact recovery of viral loads for positive samples. Nevertheless, these advantages come with significantly increased computational complexity, limiting the scalability of CS-based methods in large-scale screening.

While previous methods can effectively identify positive samples, their models do not capture the joint effects of dilution, pool size, and the $\mathrm{LOD}$ (i.e., detection threshold), which determine whether viral loads in pools are detectable. When multiple samples are combined into a pool, the viral load of a positive sample is inevitably diluted by negative samples in the same pool. If the pool size is too large, the diluted viral load may fall below the $\mathrm{LOD}$, causing pools that contain positives to be misclassified as negative \cite{arevalo2020false,arnaout2020sars,cleary2021using}. Consequently, test outcomes depend not only on the presence of positive samples, but also on the pool size and the corresponding viral loads, which cannot be captured by bit-flipping or dilution noise models in group testing. In addition, for CS-based methods, the dilution and $\mathrm{LOD}$ introduce thresholding and nonlinear distortions in the measurements, violating the assumed linear measurement model and leading to inaccurate recovery of viral loads.

Practical testing process also introduces additional uncertainty due to liquid dispensing volume errors \cite{shental2020efficient,TORRESACOSTA2022108713,FEDELE2024100195}, which perturb both the volume of each sample and the total viral load in the pool, affecting the diluted viral load in a nonlinear manner. Together with the dilution and $\mathrm{LOD}$, these factors give rise to nonlinear and pool-dependent measurements that are challenging for existing methods, limiting their performance in practical testing scenarios.

Different from prior work, we propose a novel measurement model that characterizes the actual laboratory testing process. The model formalizes the process of sample pooling and dilution, incorporates the binary quantization imposed by the $\mathrm{LOD}$, and accounts for liquid dispensing errors. While the $\mathrm{LOD}$-induced binary quantization is analogous to 1-bit compressed sensing \cite{6418031}, the measurements in our model are nonlinear and pool-dependent due to the dilution and dispensing errors. Building upon this model, we then investigate the sparse support recovery under two regimes: one with known sparsity level $k$, and the other with unknown sparsity level, which is a more challenging setting of blind support recovery. We develop low-complexity decoding algorithms for both cases and establish theoretical guarantees for exact support recovery. The main contributions of our work are summarized as follows:
\begin{enumerate}
\item{We propose a novel measurement model for group testing that captures the actual laboratory testing process, accounting for the dilution, $\mathrm{LOD}$-induced binary quantization, and liquid dispensing volume errors. In this model, a pool is declared positive only if its diluted viral load (i.e., the sum of viral loads divided by the pool size) exceeds the $\mathrm{LOD}$. Liquid dispensing errors are modeled as multiplicative perturbations drawn from a truncated Gaussian distribution, simultaneously affecting both individual sample volumes and the total viral load in the pool. Different from existing models, which typically assume either bit-flipping of test outcomes or additive noise on measurements, our model captures nonlinear, pool-dependent variations in the diluted viral loads, providing a comprehensive characterization of the testing process.}
\item{For the scenario where the sparsity level $k$ is known, we propose a low-complexity decoding algorithm based on a scoring mechanism, and establish rigorous theoretical guarantees for exact support recovery. In particular, we characterize the minimum number of tests required to achieve exact support recovery with high probability, showing that only $\mathcal{O}\left(k\log n\right)$ measurements are sufficient under both noiseless and noisy settings.}
\item{For the scenario where the sparsity level $k$ is unknown, we propose a max-gap blind support recovery algorithm, along with a heuristic variant that exploits the global pattern of the scores to identify positive samples, mitigating the impact of outliers and local fluctuations. We also establish the theoretical guarantees and show that exact support recovery can be achieved with only $\mathcal{O}\left(k\log n\right)$ measurements, even without prior knowledge of $k$.}
\item{Extensive simulations are conducted to evaluate the performance of our proposed algorithms. The results show that our methods outperform existing combinatorial group testing methods, validating the effectiveness, efficiency and robustness in large-scale screening.}
\end{enumerate}

The remainder of this paper is organized as follows. In Section~\ref{sec:Problem Formulation}, we introduce the system model and problem formulation. Section~\ref{sec:Decoding Algorithm for Known Sparsity} and Section~\ref{sec:Blind Decoding Algorithm for Unknown Sparsity} investigate the support recovery problem under the cases of known and unknown sparsity levels, respectively, where we develop the decoding algorithms and provide theoretical guarantees. In Section~\ref{sec:Numerical Experiments}, we present the simulation results. Finally, we conclude this paper in Section~\ref{sec:Conclusion}

\textit{Notation:} The set of real numbers is denoted by $\mathbb{R}$. Sets are denoted by calligraphic letters, e.g., $\mathcal{X}$, with $\left|\mathcal{X}\right|$ representing its cardinality. For a vector $x$, ${\left\| x \right\|_0}$ represents the $L_0$-norm, i.e., the number of non-zero elements in the $x$. $\mathbb{E}\left[\cdot\right]$ denotes the statistical expectation and $\mathbb{I}\left(\cdot\right)$ represents the indicator function, which equals to 1 if the condition inside is satisfied and 0 otherwise. $\mathcal{N}\left(\mu,\sigma^2\right)$ is the Gaussian distribution with mean $\mu$ and variance $\sigma^2$, and $\Phi\left(\cdot\right)$ represents the cumulative distribution function of the standard Gaussian distribution.

\section{Problem Formulation}\label{sec:Problem Formulation}
In this section, we formulate the group testing problem for identifying sparse positive samples. We begin with a noiseless model that explicitly accounts for the dilution and $\mathrm{LOD}$-induced binary quantization, and then extend it to the noisy setting by incorporating liquid dispensing volume errors.

\subsection{System Model}
Consider a population of $n$ individuals (or samples). Denote the viral load vector as $x=\left[x_1,x_2,\ldots,x_n\right]^T\in\mathbb{R}^n$, where $x_j\geq 0$ represents the viral load (measured in virus copies/mL) of the $j$-th sample, and $x_j=0$ if the $j$-th sample is uninfected. Let $\mathcal{S}$ denote the set of positive samples, i.e., the support set of $x$, defined as $\mathcal{S}=\left\{j\mid x_j>0\right\}$. In the early stage of a pandemic, the disease prevalence is extremely low, with only a few infections among thousands of individuals, implying that the vector $x$ is highly sparse, i.e., $k=\left|\mathcal{S}\right| \ll n$. Empirical studies of SARS-CoV-2 \cite{arnaout2020sars,cleary2021using,yang2021just} indicate that viral loads of infected individuals vary by many orders of magnitudes, with a very small minority of individuals harboring the vast majority of the infectious virions, which can be well approximated by a log-normal distribution \cite{yang2021just}. Accordingly, for each $j\in\mathcal{S}$, we model the viral load of the $j$-th sample as
\begin{equation}
\log_{10} {x_j} \sim \mathcal{N}\left(\mu_{x}, \sigma_x^2\right),
\end{equation}
where $\mu_x$ and $\sigma_x^2$ denote the mean and variance, respectively.

The pooling process is represented by a binary matrix $A\in\left\{0,1\right\}^{m\times n}$, with $m$ denoting the number of tests. Each row of $A$ corresponds to a pool, where $A_{i,j}=1$ indicates the $j$-th sample is assigned to the $i$-th pool, and $A_{i,j}=0$ otherwise. In this work, the pooling is constructed by independently assigning each sample to a pool with probability $p\in \left(0,\frac{1}{2}\right)$, so that each element $A_{i,j}$ equals to 1 with probability $p$ and 0 with probability $1-p$. 

Assuming that equal volumes are drawn from each sample, the viral load in the $i$-th pool is given by 
\begin{equation}\label{eq:noiseless_viral_load}
z_i=\frac{{\sum\limits_{j = 1}^n {{A _{i,j}}{x_j}} }}{{\sum\limits_{j = 1}^n {{A _{i,j}}} }},\quad i=1,\ldots,m.
\end{equation}
In practice, most diagnostic assays only detect viral material above a specific limit of detection ($\mathrm{LOD}$) and report binary outcomes (i.e., positive/negative). Due to the dilution effect caused by sample pooling, the viral load in the pool may fall below the $\mathrm{LOD}$,  even when the positives are included in the pool. Consequently, the observed binary measurement vector $y=\left[y_1,y_2,\ldots,y_m\right]^T$ is defined as
\begin{equation}\label{eq:noiseless_model}
y_i=\mathbb{I}\left(z_i\geq \mathrm{LOD}\right),\quad i=1,\ldots,m,
\end{equation}
where $\mathbb{I}\left(\cdot\right)$ is the indicator function.

The model in~(\ref{eq:noiseless_viral_load})-(\ref{eq:noiseless_model}) assumes perfect manual or robotic liquid handling. In practice, however, the volume of each sample contributed to a pool inevitably deviates from its intended value due to liquid handling errors, which are often modeled as Gaussian noise \cite{shental2020efficient}. Accordingly, we introduce a relative volume error $\epsilon_{i,j}$ for the $j$-th sample in the $i$-th pool. Since liquid volumes must be non-negative and mechanical errors are strictly bounded, we assume that $\epsilon_{i,j}$ follow a Gaussian distribution $\mathcal{N}\left(0,\sigma_{\epsilon}^2\right)$ truncated to the interval $\left[-\epsilon_{max}, \epsilon_{max}\right]$ with $\epsilon_{max}<1$.  Therefore, the noisy viral load in the $i$-th pool, denoted by $\tilde{z}_i$, is given by 
\begin{equation}\label{eq:noisy_viral_load}
  \tilde{z}_i=\frac{{\sum\limits_{j = 1}^n {\left( {1 + {\epsilon_{i,j}}} \right){A _{i,j}}{x_j}} }}{{\sum\limits_{j = 1}^n {\left( {1 + {\epsilon_{i,j}}} \right){A _{i,j}}} }}, \quad i=1,\ldots,m.
\end{equation}
The corresponding binary measurement vector $\tilde{y}=\left[\tilde{y}_1,\tilde{y}_2,\ldots,\tilde{y}_m\right]^T$ is obtained as
\begin{equation}\label{eq:noisy_model}
  \tilde{y}_i=\mathbb{I}\left(\tilde{z_i}\geq \mathrm{LOD}\right), \quad i=1,\ldots,m.
\end{equation}

It is worth noting that the liquid dispensing volume errors affect both the numerator (total viral load) and the denominator (total liquid volume) in (\ref{eq:noisy_viral_load}), which fundamentally distinguishes our model from 1-bit compressed sensing with additive noise. Moreover, when $\sigma_{\epsilon}=0$, we have $\epsilon_{i,j}=0$ for all $(i,j)$, and (\ref{eq:noisy_viral_load}) reduces to (\ref{eq:noiseless_viral_load}). Therefore, we focus on the noisy setting without loss of generality.

\subsection{Problem Formulation}
Given the pooling matrix $A\in\left\{0,1\right\}^{m\times n}$ and the binary test outcomes $\tilde{y}\in\left\{0,1\right\}^m$, our goal is to develop a robust, low-complexity decoding algorithm to identify positive samples, i.e., recover the support set $\mathcal{S}$. We investigate the support recovery problem under two scenarios: (i) the sparsity level $k$ is known a priori, i.e., the number of positives (or disease prevalence) is available; (ii) the sparsity level $k$ is unknown. Furthermore, we aim to establish rigorous lower bounds on the number of tests $m$ required for exact support recovery with high probability under both noiseless and noisy settings.

\section{Decoding Algorithm for Known Sparsity Level}\label{sec:Decoding Algorithm for Known Sparsity}
In this section, we consider the scenario where the number of positive samples, $k=\left|\mathcal{S}\right|$, is known a priori. We first propose a low-complexity decoding algorithm, and then derive the lower bounds on the required number of tests under both noiseless and noisy settings.

\subsection{Score-Based Support Recovery Algorithm}
The decoding algorithm is based on the score computed for each sample. Let $\mathcal{P^{+}}=\left\{i\mid \tilde{y}_i=1\right\}$ and $\mathcal{P^{-}}=\left\{i\mid \tilde{y}_i=0\right\}$ denote the sets of the positive and negative pools, respectively. For each sample $j\in\left\{1,2,\ldots,n\right\}$, we define its score $X_j$ as
\begin{equation}\label{eq:score}
  X_j=a \sum_{i \in \mathcal{P}^+} A_{i,j} - b \sum_{i \in \mathcal{P}^-}A_{i,j},
\end{equation}
where $a>0$ and $b>0$ are predefined reward and penalty parameters, respectively. 

To justify this reward-penalty design, we analyze the expected score for positive and negative samples under an idealized setting with a negligible detection threshold ($\mathrm{LOD}\to 0$) and no liquid dispensing errors ($\sigma_{\epsilon}=0$). For a positive sample $j\in\mathcal{S}$, the expected score is 
\begin{equation}
  \begin{aligned}
  &\mathbb{E}\left[X_j\mid j\in \mathcal{S}\right] \\
  & = a\sum_{i=1}^{m} \mathbb{P}\left(i\in \mathcal{P}^+\mid j\in \mathcal{S}, A_{i,j}=1\right)\mathbb{P}\left(A_{i,j}=1\right) \\
  &\quad - b\sum_{i=1}^{m} \mathbb{P}\left(i\in \mathcal{P}^-\mid j\in \mathcal{S}, A_{i,j}=1\right)\mathbb{P}\left(A_{i,j}=1\right)\\
  & = amp.
  \end{aligned}
\end{equation}
For a negative sample $j\notin \mathcal{S}$, its expected score is
\begin{equation}
  \begin{aligned}
  &\mathbb{E}\left[X_j\mid j\notin \mathcal{S}\right] \\
  & = a\sum_{i=1}^{m} \mathbb{P}\left(i\in \mathcal{P}^+\mid j\notin \mathcal{S}, A_{i,j}=1\right)\mathbb{P}\left(A_{i,j}=1\right) \\
  &\quad - b\sum_{i=1}^{m} \mathbb{P}\left(i\in \mathcal{P}^-\mid j\notin \mathcal{S}, A_{i,j}=1\right)\mathbb{P}\left(A_{i,j}=1\right)\\
  & = amp\left(1-(1-p)^k\right)-bmp(1-p)^k,
  \end{aligned}
\end{equation}
which is strictly smaller than the expected score of positive samples.
This simple analysis provides intuition for the proposed scoring mechanism, indicating that positive samples tend to obtain higher scores than negative samples under ideal conditions. A more rigorous analysis that accounts for the detection threshold ($\mathrm{LOD}$) and liquid dispensing errors is presented in the following section.

Based on the scoring mechanism in~(\ref{eq:score}), the decoding procedure is straightforward. After computing the scores for all samples, we sort them in descending order and select the top $k$ samples as the estimated support set $\hat{\mathcal{S}}$. The detailed procedure is given in Algorithm~\ref{score-based algorithm}.

The computational complexity of the score-based algorithm is $\mathcal{O}\left(mn+n\log n\right)$. Computing the sample scores and sorting them require $\mathcal{O}(mn)$ and $\mathcal{O}(n\log n)$ operations, respectively. Thus, the overall complexity scales nearly linearly with the numbers of tests and samples, making the algorithm suitable for large-scale screening.

\begin{algorithm}
\caption{Score-Based Support Recovery Algorithm}\label{score-based algorithm}
\begin{algorithmic}[1]
  \REQUIRE Binary measurement vector $\tilde{y}\in\left\{0,1\right\}^m$, pooling matrix $A\in\left\{0,1\right\}^{m\times n}$, the number of positives $k$, parameters $a,b>0$.
  \ENSURE Estimated support set $\hat{\mathcal{S}}$.

  \STATE Initialize scores: $X_j\leftarrow 0$ for $j=1,2,\ldots,n$;
  \FOR{$i=1,2,\ldots,m$}
    \FOR{$j\in\left\{l\mid A_{i,l}=1\right\}$}
        \IF{$\tilde{y}_i=1$}
          \STATE $X_j\leftarrow X_j+a$;
        \ELSE
          \STATE $X_j\leftarrow X_j-b$;
        \ENDIF
    \ENDFOR
  \ENDFOR
  \STATE Sort samples in descending order according to scores, and obtain an ordered sequence ($j_1,j_2,\ldots,j_n$) such that $X_{j_1}\geq X_{j_2}\geq\cdots \geq X_{j_n}$;
  \STATE $\hat{\mathcal{S}}\leftarrow \left\{j_1,j_2,\ldots,j_k\right\}$;
  \RETURN $\hat{\mathcal{S}}$.
\end{algorithmic}
\end{algorithm}

Beyond its computational efficiency, the proposed algorithm offers several advantages for identifying positive samples. Existing combinatorial group testing algorithms, such as COMP, DD, and SCOMP, rely on strict elimination rules, where any sample appearing in a negative pool is definitively excluded. However, in our measurement model, the dilution effect and dispensing errors may cause the viral loads of pools containing positives to fall below the $\mathrm{LOD}$, resulting in false negatives. In such cases, strict elimination rules risk missing true positive samples. In contrast, our algorithm assigns each sample a score which leverages information from all tests, mitigating false negatives and improving recovery robustness. Moreover, our algorithm relies only on the binary test outcomes $\tilde{y}$ and the pooling matrix $A$, without requiring knowledge of the actual liquid dispensing volumes. Thus, volume calibration is not required during decoding, facilitating practical implementation.

\subsection{Theoretical Performance Guarantees}
For each pool $i\in\left\{1,\ldots,m\right\}$, we define the conditional probabilities
\begin{align}
  \tilde{p}_{E_1}&\triangleq \mathbb{P}\left(\tilde{y}_i=1\mid j\in \mathcal{S},A_{i,j}=1\right), \\
  \tilde{p}_{E_2}&\triangleq \mathbb{P}\left(\tilde{y}_i=1\mid j\notin \mathcal{S},A_{i,j}=1\right),
\end{align}
which represent the probabilities that a pool tests positive given that it contains a specific positive or negative sample, respectively. The difference between $\tilde{p}_{E_1}$ and $\tilde{p}_{E_2}$ is denoted as 
\begin{equation}
  \Delta \tilde{p}\triangleq \tilde{p}_{E_1}-\tilde{p}_{E_2},
\end{equation}
which measures how reliably positive samples can be distinguished from negative ones based on test outcomes. When $\mathrm{LOD}$ is sufficiently small and $\sigma_{\epsilon}=0$ in the model~(\ref{eq:noisy_viral_load})-(\ref{eq:noisy_model}), the corresponding probability gap is strictly positive. However, in practice, this condition is not guaranteed due to the dilution and dispensing errors, which may cause pools containing positives to test negative, reducing $\tilde{p}_{E_1}$ and even making it smaller than $\tilde{p}_{E_2}$.

To ensure the reliable distinguishability between positive and negative samples, the effects of dilution and dispensing errors must be controlled. The following lemma provides sufficient conditions under which the probability gap $\Delta \tilde{p}$ remains positive and is lower bounded by $\left(1-p\right)^k / 2$. A larger $\Delta \tilde{p}$ implies stronger statistical separation between positive and negative samples, thereby improving the reliability of the score-based algorithm. The corresponding probability gap in the noiseless setting is denoted by $\Delta p$.

\begin{lemma}\label{lemma:probability gap}
  Consider the system model in~(\ref{eq:noisy_viral_load})-(\ref{eq:noisy_model}), where $x$ is a $k$-sparse vector and $A$ is a Bernoulli random matrix with $p\in\left(0,\frac{1}{2}\right)$. The viral loads of positive samples follow a log-normal distribution with mean $\mu_x$ and variance $\sigma_x^2$. The liquid dispensing volume errors follow a Gaussian distribution $\mathcal{N}\left(0,\sigma_{\epsilon}^2\right)$ truncated to the interval $\left[-\epsilon_{max}, \epsilon_{max}\right]$. If the probability $p$ satisfies
  \begin{equation}
    0< p \leq \frac{2\left(d_T-1\right)\Phi_{d_T}-d_T}{2\left(n-1\right)\left(1+{\sigma_{\epsilon}^2}/\left({1-\epsilon_{max}^2}\right)\right)\Phi_{d_T}+kd_T},
  \end{equation}
  where ${\Phi _{{d_T}}} = \Phi \left( {{\left({{\mu _x} - {{\log }_{10}}\left( {{\mathrm{LOD}} \cdot {d_T}} \right)}\right)}/{{{\sigma _x}}}} \right)$ and $d_T$ satisfies
  \begin{equation}
    \frac{1}{1-\alpha}< d_T <\frac{1}{\mathrm{LOD}}10^{\mu_x-\sigma_x\Phi^{-1}\left({1}/{\left(2\alpha\right)}\right)},\alpha\in\left(\frac{1}{2},1\right),
  \end{equation}
  then we have 
  \begin{equation}
    \Delta \tilde{p} \geq \frac{\left(1-p\right)^k}{2}.
  \end{equation}
  In the noiseless setting, the probability gap satisfies
  \begin{equation}
    \Delta p=\Delta \tilde{p} + \frac{\left(n-1\right)p\sigma_{\epsilon}^2}{\left(1-\epsilon_{max}^2\right)d_T}\Phi_{d_T}.
  \end{equation}
\end{lemma}

\begin{proof}
  See Appendix~\ref{appendix:lemma_probability gap}.
\end{proof}

In the noiseless setting without liquid dispensing errors, the probability gap $\Delta p$ is larger than $\Delta\tilde{p}$, indicating that dispensing errors reduce the distinguishability between positive and negative samples, making it more challenging to identify positives.

With Lemma~\ref{lemma:probability gap}, we are now ready to provide a sufficient condition on the number of tests required for exact support recovery.
\begin{theorem}[Known Sparsity Level]\label{theorem:score-based algorithm}
  Under the conditions of Lemma~\ref{lemma:probability gap}, let $\hat{\mathcal{S}}$ be the estimated support set of Algorithm~\ref{score-based algorithm} with parameters $a,b>0$. Then, for arbitrarily small $\delta>0$, we have
  \begin{equation}
    \mathbb{P}\left(\hat{\mathcal{S}}=\mathcal{S}\right)\geq 1-\delta,
  \end{equation}
  provided that the number of tests satisfies
  \begin{equation}
    m \geq \frac{4 \left(2r^2 + {\Delta \tilde{p}}r/3\right) \log\left(n/\delta\right)}{p \left({\Delta \tilde{p}}\right)^2},
  \end{equation}
  where $r=\frac{\max\left(a, b\right)}{a+b}$.
\end{theorem}

\begin{proof}
  See Appendix~\ref{appendix:theorem_score-based algorithm}.
\end{proof}

\begin{remark}
  Theorem~\ref{theorem:score-based algorithm} shows that the recovery performance of the score-based algorithm is governed by the probability gap $\Delta \tilde{p}$. Using the lower bound $\Delta \tilde{p} \geq {\left(1-p\right)^k}/{2}$ derived in Lemma~\ref{lemma:probability gap}, the required number of tests scales as
  \begin{equation}\label{eq:score_m}
    m=\mathcal{O}\left(\frac{\log\left(n/\delta\right)}{p\left(1-p\right)^{2k}}\right),
  \end{equation}
  which grows logarithmically with the number of samples $n$. For fixed $n$, $k$, and $\delta$, $m$ is minimized when $p=1/\left(2k+1\right)$. Substituting this value into~(\ref{eq:score_m}), shows that the number of tests scales approximately as $\mathcal{O}\left(k\log n\right)$.
\end{remark}

In the noiseless setting, the probability gap $\Delta p$ is larger than $\Delta \tilde{p}$, implying that fewer tests are required to achieve the same recovery performance compared to the noisy setting. Nevertheless, even in the presence of dilution effect and dispensing errors, exact support recovery can still be achieved with a number of tests scaling approximately as $\mathcal{O}\left(k\log n\right)$, demonstrating the effectiveness and robustness of our algorithm for large-scale screening applications.

\section{Blind Decoding Algorithm for Unknown Sparsity Level}\label{sec:Blind Decoding Algorithm for Unknown Sparsity}
In this section, we propose a blind decoding algorithm to recover the support set without any prior knowledge, and establish sufficient conditions on the number of tests required for exact support recovery.
\subsection{Max-Gap Blind Support Recovery Algorithm}
Without prior knowledge of the sparsity level $k$, the top-$k$ sample selection strategy employed in Algorithm~\ref{score-based algorithm} is no longer applicable. We propose a max-gap decoding algorithm, which can directly estimate the number of positives from the distribution of the scores. The key insight is as follows. The scoring mechanism in~(\ref{eq:score}) exhibits a clear separation in the ordered score sequence, since the score of positive samples tends to be substantially higher than that of negative samples. By detecting the largest gap between consecutive sorted scores, the algorithm can estimate the sparsity level $\hat{k}$ and select the most likely positive samples.

The proposed max-gap algorithm proceeds as follows. It first computes the score for all samples according to~(\ref{eq:score}) and sorts in descending order. Next, the algorithm calculates the differences between consecutive sorted scores and select $\hat{k}$ corresponding to the largest gap, which cannot exceed a predefined maximum sparsity level. The estimated support set is taken as the top $\hat{k}$ samples. The detailed procedure is summarized in Algorithm~\ref{max-gap algorithm}.
\begin{algorithm}
\caption{Max-Gap Blind Support Recovery Algorithm}\label{max-gap algorithm}
\begin{algorithmic}[1]
  \REQUIRE Binary measurement vector $\tilde{y}\in\left\{0,1\right\}^m$, pooling matrix $A\in\left\{0,1\right\}^{m\times n}$, maximum sparsity level $K$, parameters $a,b>0$.
  \ENSURE Estimated support set $\hat{\mathcal{S}}$.

  \STATE Initialize scores: $X_j\leftarrow 0$ for $j=1,2,\ldots,n$;
  \FOR{$i=1,2,\ldots,m$}
    \FOR{$j\in\left\{l\mid A_{i,l}=1\right\}$}
        \IF{$\tilde{y}_i=1$}
          \STATE $X_j\leftarrow X_j+a$;
        \ELSE
          \STATE $X_j\leftarrow X_j-b$;
        \ENDIF
    \ENDFOR
  \ENDFOR
  \STATE Sort samples in descending order according to scores, and obtain an ordered sequence ($j_1,j_2,\ldots,j_n$) such that $X_{j_1}\geq X_{j_2}\geq\cdots \geq X_{j_n}$;
  \STATE Compute gaps: $\Delta X_l=X_{j_l}-X_{j_{l+1}}$ for $l=1,2,\ldots,n-1$;
  \STATE $\hat{k}=\arg \max\limits_l \left\{\Delta X_{l}: l\leq K\right\}$;
  \STATE $\hat{\mathcal{S}}\leftarrow \left\{j_1,j_2,\ldots,j_{\hat{k}}\right\}$;
  \RETURN $\hat{\mathcal{S}}$.
\end{algorithmic}
\end{algorithm}

The max-gap algorithm exploits the same statistical mechanism that positive samples tend to attain higher scores than negative ones due to appearing more frequently in positive pools. In addition, by integrating information from all tests, the algorithm is robust to false negatives caused by the dilution effect and liquid dispensing errors. Overall, the max-gap algorithm preserves the robustness of the score-based approach and enables reliable support recovery even when the sparsity level is unknown, making it suitable for practical large-scale screening scenarios.

\subsection{Heuristic Max-Gap Blind Support Recovery Algorithm}
While the max-gap algorithm identifies positive samples by detecting the largest gap in the ordered score sequence, its performance may be affected by isolated large gaps and local fluctuations. In particular, unusually large gaps may appear far beyond the true sparsity level, leading to overestimation of $k$, while local fluctuations may obscure the separation between positive and negative samples. Nevertheless, the max-gap algorithm is simple and effective when the score sequence exhibits a clear separation between positive and negative samples.

We propose a heuristic variant that exploits the global pattern of the ordered score sequence. In the region corresponding to positive samples, the gaps vary noticeably, whereas beyond this region, most gaps become very small and remain consistently close to zero. By identifying the point after which the gaps stay small over a range, the heuristic variant provides a more robust estimate of the end of the positive samples, particularly in the presence of outliers or noise. The detailed procedure is given in Algorithm~\ref{heuristic max-gap algorithm}. Compared with the max-gap algorithm, the heuristic variant is less sensitive to local fluctuations and anomalous gaps, although it incurs slightly higher computational overhead and is more challenging to analyze theoretically.
\begin{algorithm}
\caption{Heuristic Max-Gap Algorithm}\label{heuristic max-gap algorithm}
\begin{algorithmic}[1]
  \REQUIRE Binary measurement vector $\tilde{y}\in\left\{0,1\right\}^m$, pooling matrix $A\in\left\{0,1\right\}^{m\times n}$, maximum sparsity level $K$, parameters $a,b>0$.
  \ENSURE Estimated support set $\hat{\mathcal{S}}$.
  
  \STATE Compute scores $\left\{X_j\right\}_{j=1}^n$ for all samples;
  \STATE Sort samples in descending order according to scores, and obtain an ordered sequence ($j_1,j_2,\ldots,j_n$) such that $X_{j_1}\geq X_{j_2}\geq\cdots \geq X_{j_n}$;
  \STATE Compute gaps: $\Delta X_l=X_{j_l}-X_{j_{l+1}}$ for $l=1,2,\ldots,n-1$;
  \STATE Choose candidates corresponding to relatively large gap: \\$\mathcal{C}=\left\{l\mid \Delta X_l \geq \alpha \max \limits_{t\leq K}{\Delta X_t}, \alpha\in\left(0,1\right)\right\}$;
  \FOR{$l\in\mathcal{C}$}
    \STATE Compute the near-zero gaps between candidates: \\$F(l)=\sum\nolimits_{t = l + 1}^{l^{+}} {\mathbb{I}\left( {\Delta {X_t} \leq \epsilon } \right)}$, where $l^{+}$ is the next candidate in $\mathcal{C}$ (or $n$ if none), and $\epsilon$ is the small threshold;
  \ENDFOR
  \STATE $\hat{k}=\arg\max \limits_{l\in\mathcal{C},l\leq K} \frac{F(l)+1}{F(l^-)+1}$, where $l^{-}$ is the previous candidate in $\mathcal{C}$;
  \STATE $\hat{\mathcal{S}}\leftarrow \left\{j_1,j_2,\ldots,j_{\hat{k}}\right\}$;
  \RETURN $\hat{\mathcal{S}}$.
\end{algorithmic}
\end{algorithm}

\subsection{Theoretical Performance Guarantees}
We now provide the sufficient condition on the number of tests required to achieve exact support recovery using the max-gag algorithm.
\begin{theorem}[Unknown Sparsity Level]\label{theorem:max-gap algorithm}
  Under the conditions of Lemma~\ref{lemma:probability gap}, let $\hat{\mathcal{S}}$ be the estimated support set of Algorithm~\ref{max-gap algorithm} with parameters $a,b>0$. Then, for arbitrarily small $\delta>0$, we have
  \begin{equation}
    \mathbb{P}\left(\hat{\mathcal{S}}=\mathcal{S}\right)\geq 1-\delta,
  \end{equation}
  provided that the number of tests satisfies
  \begin{equation}
    m \geq \frac{16 \left(2r^2 + {\Delta \tilde{p}}r/6\right) \log\left(2n/\delta\right)}{p \left({\Delta \tilde{p}}\right)^2},
  \end{equation}
  where $r=\frac{\max\left(a, b\right)}{a+b}$.
\end{theorem}

\begin{proof}
  See Appendix~\ref{appendix:theorem_max-gap algorithm}.
\end{proof}

\begin{remark}
  Theorem~\ref{theorem:max-gap algorithm} shows that the performance of the max-gap algorithm is also governed by the probability gap $\Delta \tilde{p}$. Applying the lower bound $\Delta \tilde{p} \geq {\left(1-p\right)^k}/{2}$ derived in Lemma~\ref{lemma:probability gap}, the required number of tests scales as
  \begin{equation}
    m=\mathcal{O}\left(\frac{\log\left(n/\delta\right)}{p\left(1-p\right)^{2k}}\right).
  \end{equation}
  The asymptotic scaling with respect to $n$, $k$ and $p$ is the same as that of the score-based algorithm. The primary difference lies in a larger multiplicative constant required to compensate for the uncertainty of the unknown sparsity level.
\end{remark} 

It is difficult to provide theoretical guarantees for the heuristic variant. Nevertheless, the heuristic algorithm is less sensitive to isolated large gaps and local fluctuations. It can achieve more accurate and robust recovery in practical testing scenarios, as will be illustrated in the simulation results, outperforming the max-gap algorithm in most cases.

When the sparsity level $k$ is known a priori, the score-based algorithm only needs to select the top $k$ samples with the highest scores. In this case, exact support recovery is guaranteed as long as the minimum score among the positive samples exceeds the maximum score among the negative samples, even if the gap between them is small. In contrast, when the sparsity level $k$ is unknown, the decoding algorithm must identify the boundary between positive and negative samples without any prior knowledge. This demands a larger number of tests to ensure sufficient separation between the scores, so that a clear maximum gap appears in the ordered score sequence. Consequently, compared to the score-based algorithm, the max-gap algorithm generally requires more tests to achieve the same recovery performance.

\begin{figure*}[!t]
\centering
\includegraphics[width=\textwidth]{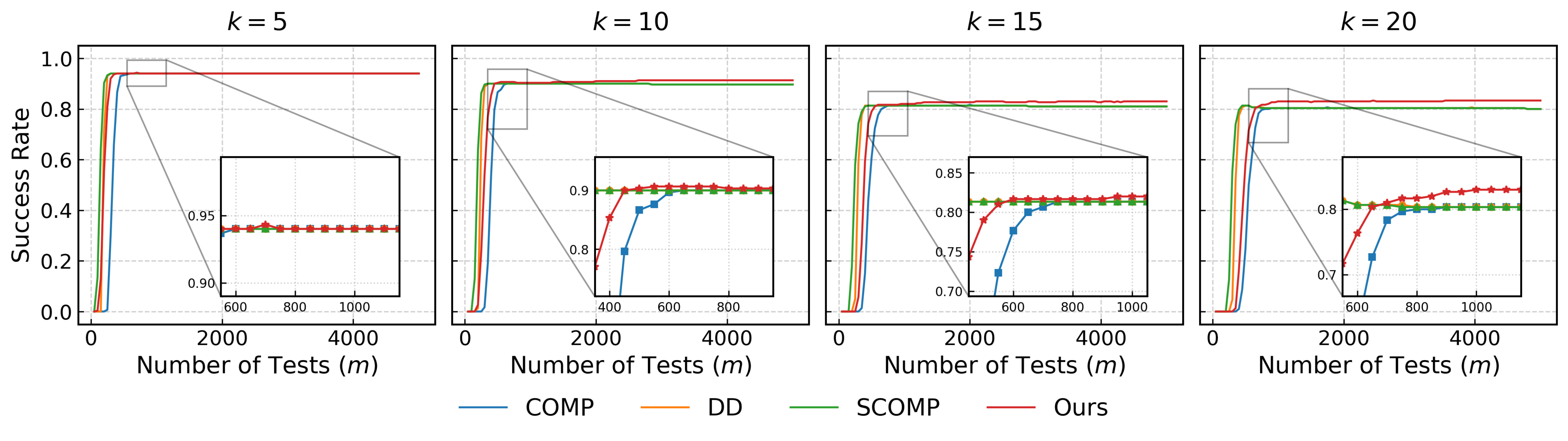}
\vspace{-0.8cm}
\caption{Performance comparison of different algorithms under varying sparsity levels, $k\in\left\{5,10,15,20\right\}$, with $p=\frac{1}{32}$ and $\mathrm{LOD}=1$. The parameters of the score-based algorithm are set to $a=1$ and $b=10$.}
\label{knonwK_success_rate_k}
\vspace{-0.4cm}
\end{figure*}

\section{Numerical Experiments}\label{sec:Numerical Experiments}
In this section, we conduct numerical simulations to evaluate the performance of the proposed algorithms, demonstrating their effectiveness and robustness under different sparsity levels and constraints.

\subsection{Experimental Setup}
We consider a population of size $n=5000$ with $k$ infected individuals, which are selected uniformly at random. The number of infections $k$ takes values in $\left\{5,10,15,20\right\}$. Following \cite{yang2021just}, the distribution of viral load can be well approximated by a log-normal distribution, i.e., for each $j\in \mathcal{S}$, $\log_{10}x_j \sim \mathcal{N}\left(\mu_x, \sigma_x^2\right)$. The parameters are set to $\mu_x=6.39$ and $\sigma_x^2=1.83^2$, based on the findings in \cite{abi5273}. The liquid dispensing volume errors are modeled as truncated Gaussian noise, i.e., $\epsilon_{i,j}\sim \mathcal{N}\left(0,\sigma_{\epsilon}^2\right)$ truncated to the interval $\left[-\epsilon_{max},\epsilon_{max}\right]$, where $\epsilon_{max}=0.1$ and $\sigma_{\epsilon}=0.01$. 

The pooling matrix $A\in\left\{0,1\right\}^{m\times n}$ is generated according to the Bernoulli distribution, where each element $A_{i,j}$ equals to 1 with probability $p$ and 0 otherwise. For each parameter configuration, a series of $m$ (ranging from 1 to $n$) are simulated over 300 independent trials. In each trial, both the viral load vector $x$ and the pooling matrix $A$ are randomly generated, and the binary measurement vector $\tilde{y}$ is obtained according to~(\ref{eq:noisy_viral_load})-(\ref{eq:noisy_model}), with the $\mathrm{LOD}$ chosen from $\left\{1,10,50,100\right\}$ \cite{fung2020direct}. A trial is considered successful if all positive samples are correctly identified, i.e., the true positive rate equals to 1, and the success rate is the fraction of successful trials among the total trials. We also adopt two widely used metrics, recall and precision. Recall measures the fraction of true positives that are correctly identified, and precision measures the fraction of predicted positive samples that are true positives.

In the experiments, we compare our algorithms with existing combinatorial group testing algorithms, including COMP, DD, SCOMP \cite{aldridge2014group}, SAFFRON\cite{8771121}, and the fast splitting algorithm in \cite{price2023fast}. For COMP, DD and SCOMP, the same pooling matrix is used. When the sparsity level $k$ is known, these algorithms output a set of $k$ samples selected from their candidate set. When $k$ is unknown, the algorithms directly output all samples identified as positive according to their respective decoding rules. For SAFFRON and the fast splitting algorithm, they are implemented using their prescribed pooling designs and corresponding decoding algorithms, where $k$ is used to design the pooling matrix or guide hierarchical splitting. Comparisons with these two algorithms are conducted only when $k$ is known.

\subsection{Experimental Results}
\subsubsection{Support Recovery with Known Sparsity Level}
We first evaluate the recovery performance of our algorithm in comparison with COMP, DD, and SCOMP, under varying sparsity levels, $k\in\left\{5,10,15,20\right\}$, with $p=\frac{1}{32}$ and $\mathrm{LOD}=1$. As shown in Fig.~\ref{knonwK_success_rate_k}, the proposed score-based algorithm outperforms classical combinatorial group testing algorithms across all sparsity levels. In particular,
compared with SCOMP, our method achieves an improvement in success rate of approximately $1.7\%\text{--}3.3\%$ across different sparsity levels, with the performance gain becoming more pronounced as $k$ increases. This demonstrates the effectiveness of the proposed method in reliably identifying positive samples under different sparsity levels.

We further compare the recovery performance and computational efficiency of our method with SAFFRON and the fast splitting algorithm under a fixed sparsity level $k=10$, with $p=\frac{1}{32}$, $\mathrm{LOD}=1$. The results of $m=1000$ and $m=2000$ are presented in Fig.~\ref{fig:combined_performance_runtime}. Our method achieves significantly higher success rates while maintaining lower computational cost. When $m=1000$, our method achieves a success rate of $88\%$, with a recall of $98.8\%$ and a precision of $98.8\%$. In comparison, SAFFRON attains a precision of $100\%$, but identifies only a fraction of true positives, resulting in a recall of $80.7\%$ and a much lower success rate of $29\%$. The fast splitting algorithm achieves a success rate ($86\%$) comparable to ours and a high recall of $98.5\%$, however, its precision drops sharply to $1.4\%$, indicating that a large number of negative samples are incorrectly declared positive. As the number of tests increases to $m=2000$, the recovery performance of all methods improves. The success rates of SAFFRON and the fast splitting algorithm increase to $82.3\%$ and $87.3\%$, respectively. Nevertheless,
our method still achieves the highest success rate ($88.7\%$) while maintaining both high recall ($98.9\%$) and high precision ($98.9\%$), demonstrating its ability to identify true positives with a low false-positive rate. In terms of running time, we compare the time required for pooling design and decoding. Experiments are conducted on a Linux server equipped with an Intel Xeon Gold 6226R CPU (32 cores). For $m=1000$, our method requires 0.062s for pooling design and 0.012s for decoding, with a total runtime of only 0.074s. Although SAFFRON exhibits the fastest decoding speed (0.001s), its pooling design requires 0.362s, leads to a total runtime nearly five times larger than that of our method. The fast splitting algorithm spends 0.129s on pooling design and 0.14s on decoding, yielding a total runtime of 0.269s, more than three times that of our method. For $m=2000$, the total runtime of our method increases slightly to 0.15s, but remains comparable to that of the fast splitting algorithm and is still lower than SAFFRON. This computational advantage stems from the simple Bernoulli pooling design and the efficient decoding procedure of our method, enabling fast and reliable identification with low computational cost.
\begin{figure}[!t]
\centering
\subfloat[Recovery performance]{
    \includegraphics[width=\linewidth]{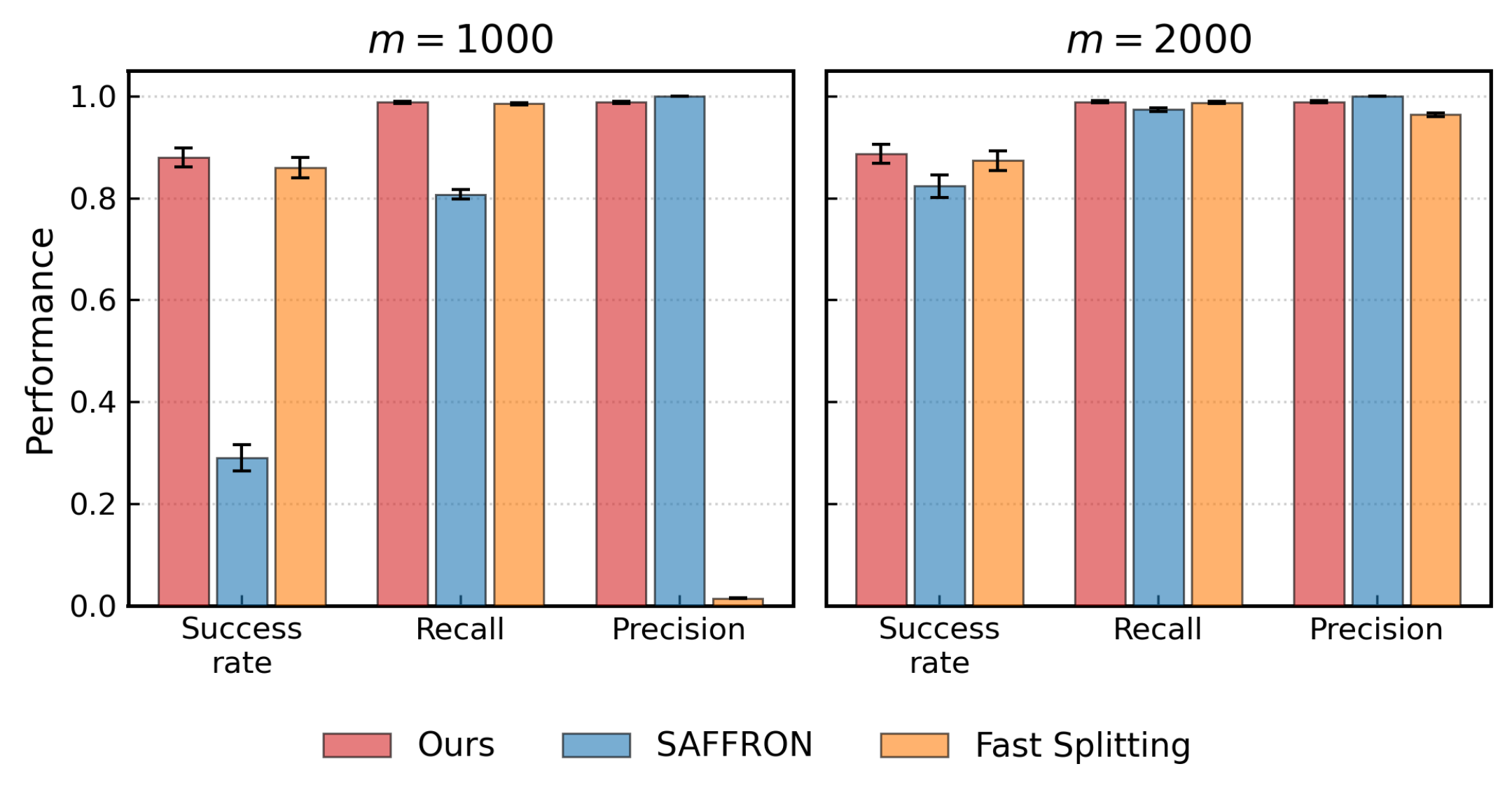}
}
\hfill
\subfloat[Running time]{
    \includegraphics[width=\linewidth]{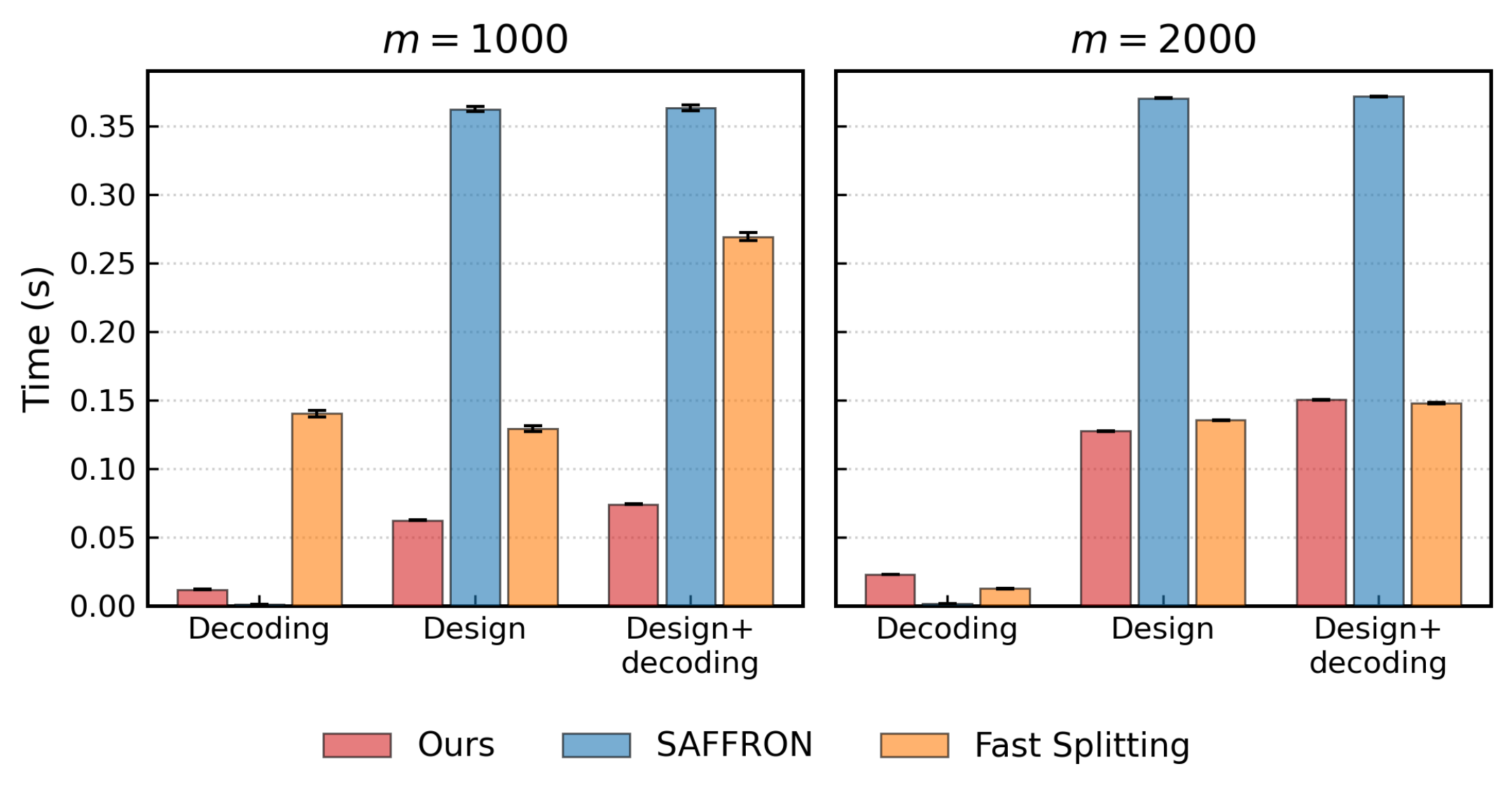}
}
\caption{Comparison of different algorithms under different numbers of tests, $m=1000$ and $m=2000$, with $k=10$, $p=\frac{1}{32}$, and $\mathrm{LOD}=1$. The parameters of the score-based algorithm are set to $a=1$ and $b=10$.}
\label{fig:combined_performance_runtime}
\vspace{-0.4cm}
\end{figure}

\begin{figure}[!t]
\centering
\includegraphics[width=\linewidth]{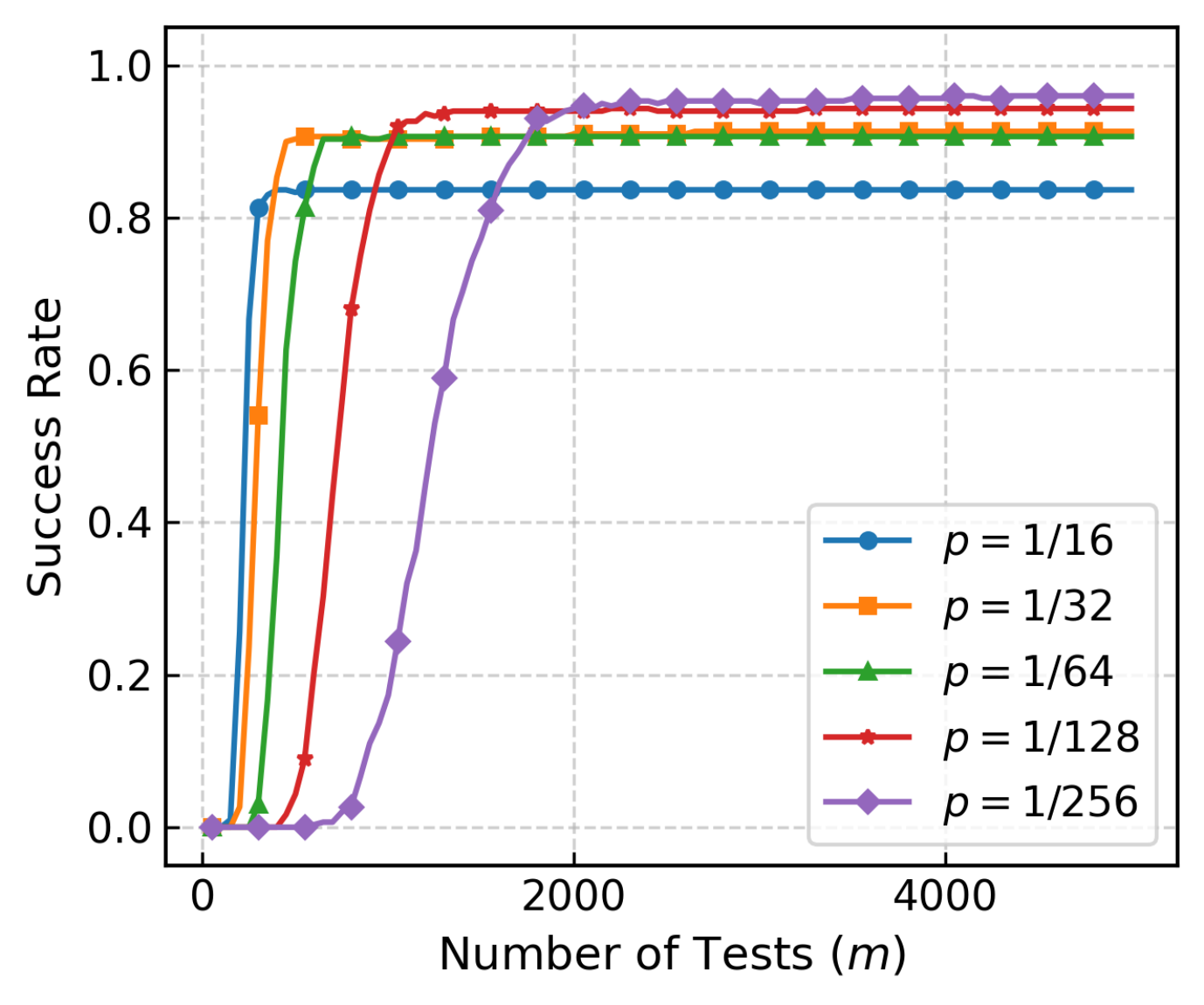}
\vspace{-0.8cm}
\caption{Recovery performance of the score-based algorithm under different pooling probabilities, $p\in\left\{\frac{1}{16},\frac{1}{32},\frac{1}{64},\frac{1}{128},\frac{1}{256}\right\}$, with $k=10$ and $\mathrm{LOD}=1$. The parameters are set to $a=1$ and $b=10$.}
\label{knonwK_success_rate_p}
\vspace{-0.4cm}
\end{figure}

To investigate the sensitivity of the algorithm to the dilution effect, we evaluate the recovery performance under different pooling probabilities $p\in\left\{\frac{1}{16},\frac{1}{32},\frac{1}{64},\frac{1}{128},\frac{1}{256}\right\}$, with $k=10$ and $\mathrm{LOD}=1$. Fig.~\ref{knonwK_success_rate_p} reveals that the choice of pooling probability $p$ has a significant impact on the recovery performance. When $p$ is small, the average pool size ($np$) remains limited, increases the likelihood that the diluted viral load exceeds the $\mathrm{LOD}$ and thereby mitigating the dilution effect. However, each sample participates in only a small number of pools, leading to insufficient information for decoding. As a result, a larger number of tests is required to achieve exact recovery. In contrast, when $p$ is large, each sample appears in many pools, providing more information and thereby reducing the number of tests required. Nevertheless, the increased pool size significantly amplifies the dilution effect, leading to more false negatives and the degradation in success rate. These results reveal a trade-off in the choice of pooling probability $p$: too small values provide insufficient information, while excessively large values lead to more severe dilution effects. Therefore, selecting an appropriate $p$ is crucial to balance information sufficiency and robustness against dilution.

\begin{figure*}[!t]
\centering
\includegraphics[width=\textwidth]{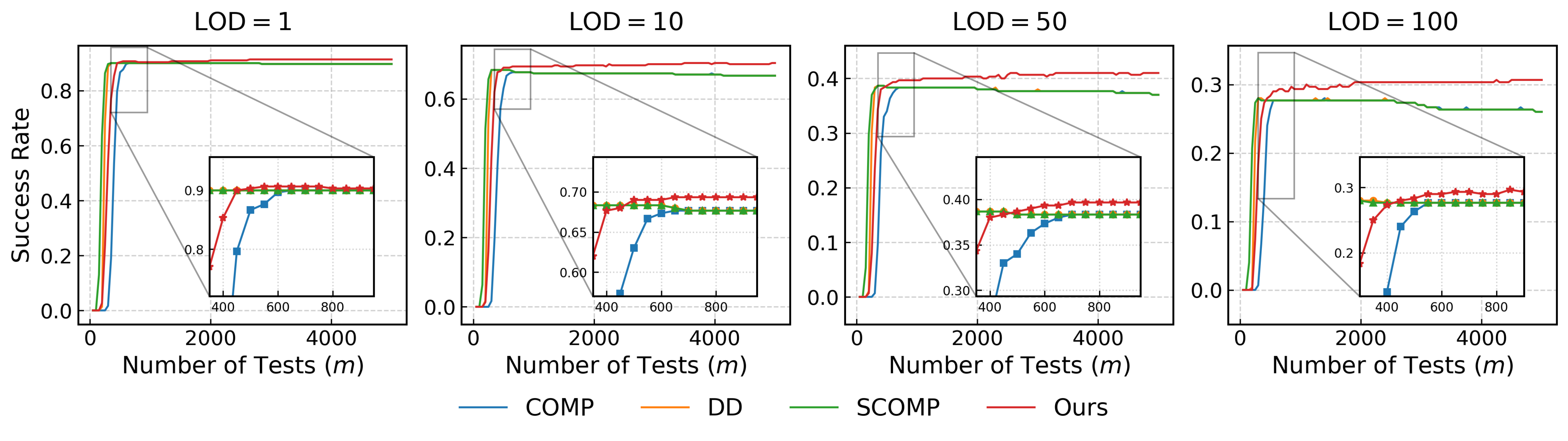}
\vspace{-0.8cm}
\caption{Recovery performance of different algorithms evaluated under different values of $\mathrm{LOD}$, with $k=10$ and $p=\frac{1}{32}$. The parameters of the scored-based algorithm are set to $a=1$ and $b=10$.}
\label{knonwK_success_rate_lod}
\vspace{-0.4cm}
\end{figure*}

\begin{figure*}[!t]
\centering
\includegraphics[width=\textwidth]{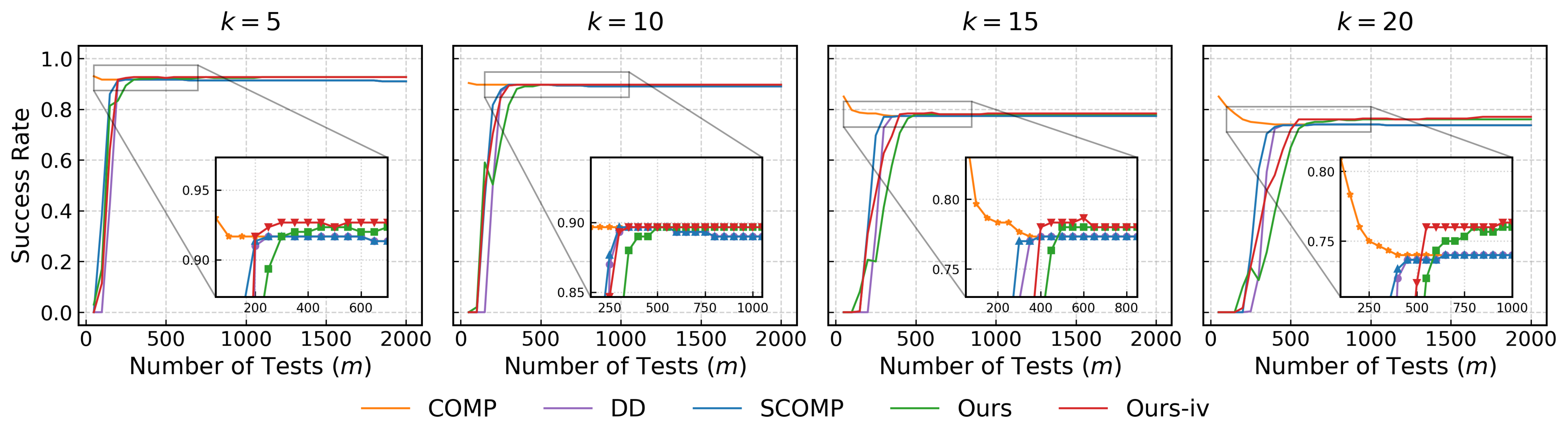}
\vspace{-0.8cm}
\caption{Performance comparison of different algorithms under varying sparsity levels, $k\in\left\{5,10,15,20\right\}$, with $p=\frac{\sqrt{2}}{32}$ and $\mathrm{LOD}=1$. The parameters of the max-gap algorithm are set to $a=1$ and $b=10$.}
\label{unknownK_success_rate_k}
\vspace{-0.4cm}
\end{figure*}
We further evaluate the impact of $\mathrm{LOD}$ on the recovery performance. Fig.~\ref{knonwK_success_rate_lod} shows the success rate of all methods under different values of $\mathrm{LOD}$, with $k=10$ and $p=\frac{1}{32}$. As the $\mathrm{LOD}$ increases, the success rate of all methods decreases, while the proposed method consistently outperforms the combinatorial group testing algorithms. Specifically, the success rate of our method decreases from $91.33\%$ at $\mathrm{LOD}=1$ to $70.33\%$ at $\mathrm{LOD}=10$, and further drops to $41\%$ and $30.67\%$ when $\mathrm{LOD}=50$ and $100$, respectively. The performance degradation is because a larger $\mathrm{LOD}$ imposes a stricter threshold for declaring a pool as positive, making it more difficult to detect diluted viral loads, particularly for pools containing samples with low viral loads. This impact is more pronounced for combinatorial group testing methods, which rely heavily on negative test outcomes to eliminate samples. As the $\mathrm{LOD}$ increases, there are more pools containing positive samples misclassified as negative, leading to the exclusion of true positive samples from the candidate set and consequently significant performance degradation. In contrast, the proposed method leverages information from all tests, making it less sensitive to such misclassifications and more robust under high $\mathrm{LOD}$. These results highlight the importance of accounting for the $\mathrm{LOD}$ in practical testing scenarios, and demonstrate that our method maintains comparatively reliable performance even as $\mathrm{LOD}$ increases.

\subsubsection{Blind Support Recovery with Unknown Sparsity Level}
In practical large-scale screening scenarios, the sparsity level $k$ is typically unknown, which poses significant challenges for exact support recovery. In this setting, we evaluate the performance of the proposed max-gap algorithm and its heuristic variant under different sparsity levels, $k\in\left\{5,10,15,20\right\}$, with $p=\frac{\sqrt{2}}{32}$ and $\mathrm{LOD}=1$. As shown in Fig.~\ref{unknownK_success_rate_k}, both algorithms achieve better performance across all sparsity levels, with the heuristic variant outperforming the max-gap algorithm. Compared with SCOMP, the max-gap algorithm improves the success rate by approximately $0.66\%\text{--}2.67\%$ across different sparsity levels, while the heuristic variant achieves the larger improvement of approximately $0.66\%\text{--}3.34\%$. It is worth noting that COMP can achieve a higher success rate than other algorithms when $m$ is small. This is because COMP eliminates samples based on negative test outcomes and declares the remaining as positive, thus achieving a high true positive rate. However, this conservative strategy also yields a large number of false positives, resulting in a high false positive rate. Overall, these results demonstrate that even without prior knowledge of the sparsity level, our algorithms can still reliably identify positive samples. 

\begin{figure}[htb]
\centering
\includegraphics[width=\linewidth]{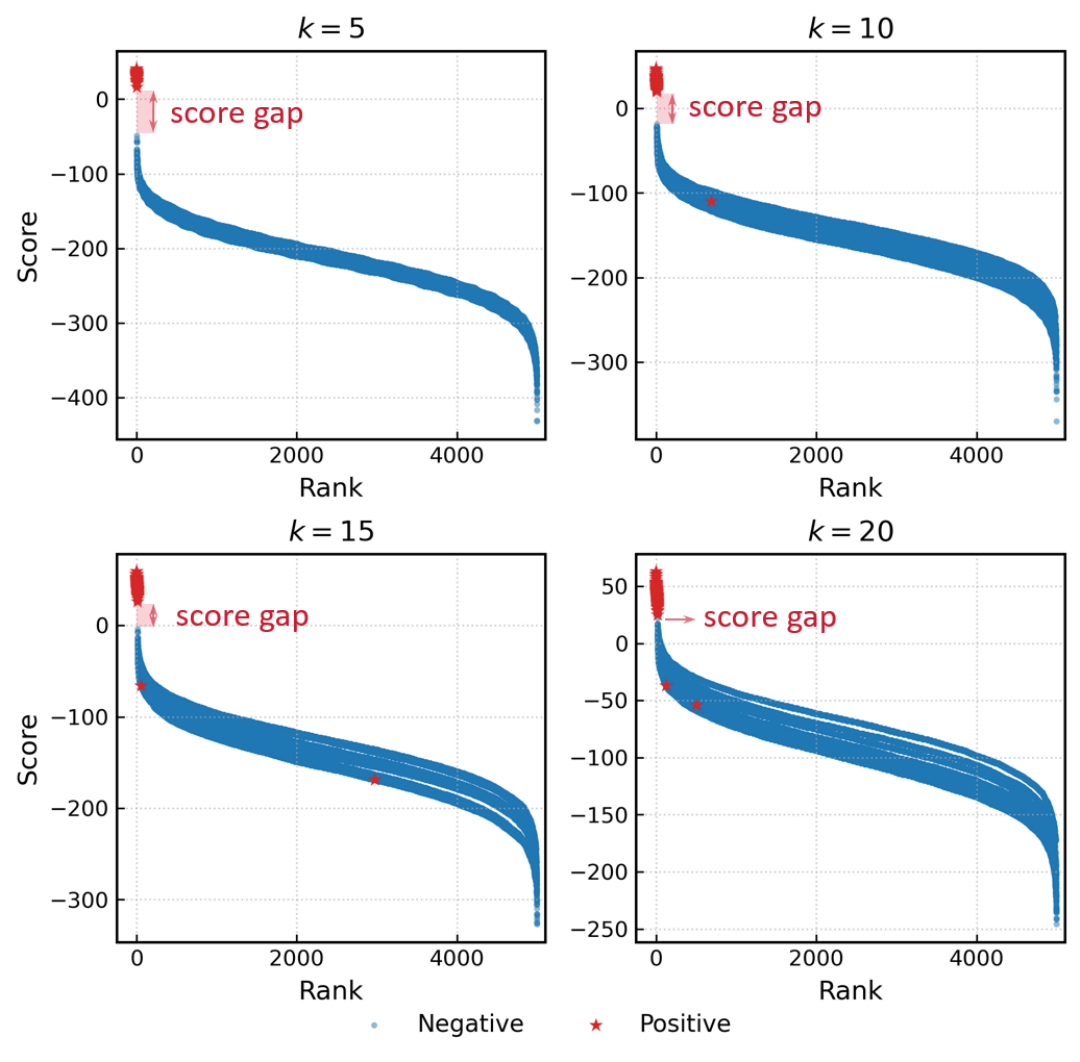}
\vspace{-0.8cm}
\caption{The sorted score distributions of 10 representative instances under a fixed number of tests, where $m=500$ for $k=5,10$ and $m=700$ for $k=15,20$, with $p=\frac{1}{16}$ and $\mathrm{LOD}=1$. The parameters are set to $a=1$ and $b=10$. Red stars indicate true positive samples, while blue dots represent negative samples.}
\label{score_sequence}
\vspace{-0.4cm}
\end{figure}
We also provide several representative examples in Fig.~\ref{score_sequence} for $k\in\left\{5,10,15,20\right\}$ with $p=\frac{1}{16}$ and $\mathrm{LOD}=1$, showing the clear separation between positive and negative samples in the ordered score sequence. In most cases, positive samples tend to have higher scores than negative ones, enabling the max-gap algorithm to effectively identify the boundary between them. However, due to the dilution effect and noise, some negative samples may occasionally attain higher scores than positive samples, leading to inaccurate estimation of sparsity level. As $k$ increases, the gap between positive and negative samples gradually narrows, and thus a larger number of tests is required to preserve a clear separation and achieve accurate recovery.

Finally, we compare the recovery performance of the max-gap algorithm and its heuristic variant under different values of $b$, with $k=10$, $p=\frac{\sqrt{2}}{32}$ and $\mathrm{LOD}=1$. The results in Fig.~\ref{unknownK_success_rate_comparison} shows that the heuristic variant outperforms the max-gap algorithm across all values of $b$, with its performance improving as $b$ increases, which is because a larger $b$ yields a clearer separation between the scores of positive and negative samples. The performance advantage of the heuristic variant stems from its ability to effectively mitigate the impact of outliers and local fluctuations in the ordered score sequence. In practice, due to the randomness in pooling design, together with the dilution effect and noise, some negative samples may occasionally attain high scores, while some positive samples have relatively low scores. This leads to isolated large gaps or local fluctuations, which can mislead the max-gap algorithm and result in an inaccurate estimation of the sparsity level. In contrast, the heuristic variant incorporates additional information about the global pattern of the ordered score sequence, rather than relying only on a single largest gap, thereby achieving more accurate estimation of the sparsity level and improving recovery performance.
\begin{figure}[!t]
\centering
\includegraphics[width=\linewidth]{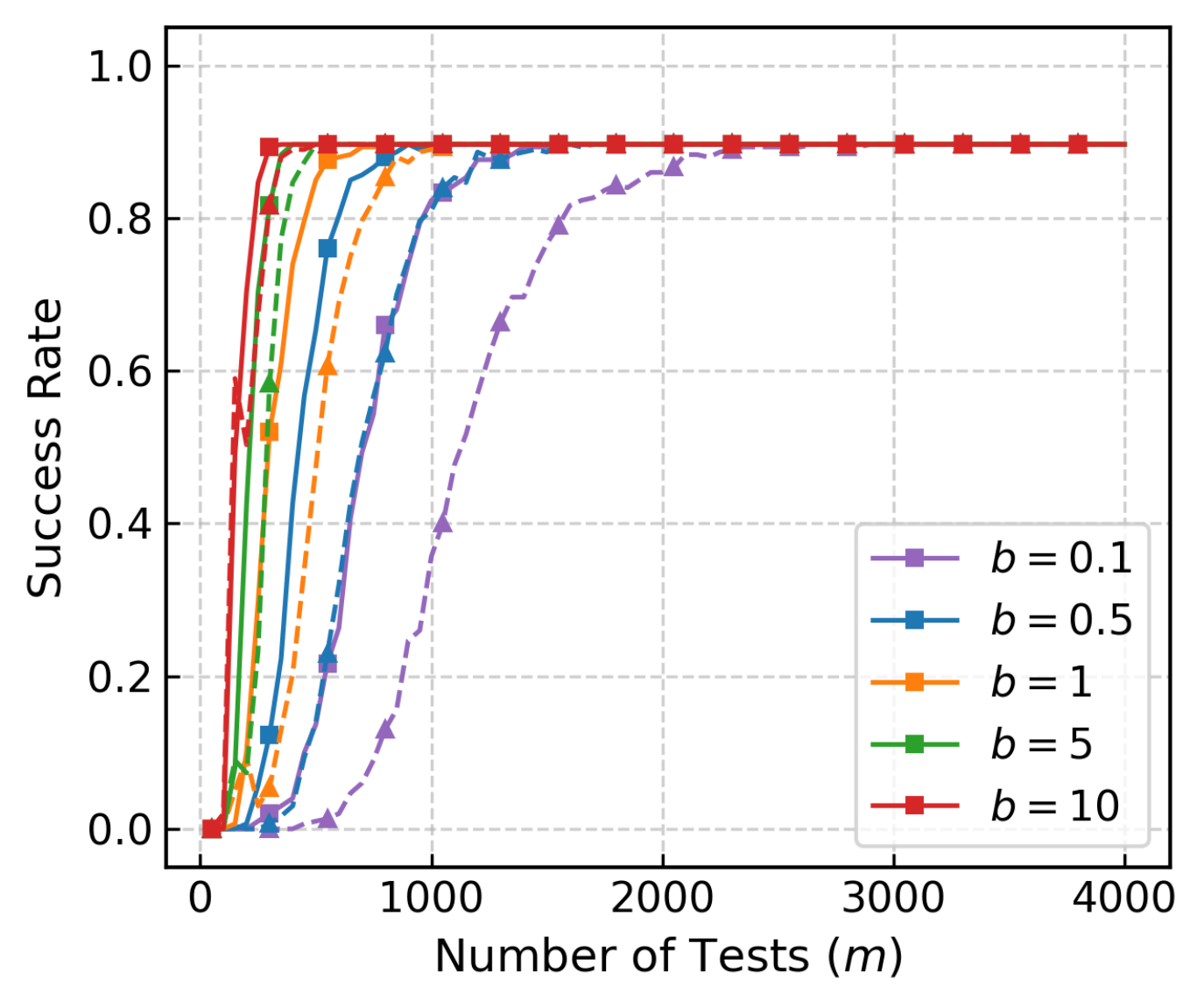}
\vspace{-0.8cm}
\caption{Performance comparison between the max-gap algorithm and its heuristic variant, with $k=10$, $p=\frac{\sqrt{2}}{32}$ and $\mathrm{LOD}=1$. The parameters are set to $a=1$ and $b\in\left\{0.1,0.5,1,5,10\right\}$. Dashed lines represent the max-gap algorithm, while solid lines represent the heuristic variant.}
\label{unknownK_success_rate_comparison}
\vspace{-0.4cm}
\end{figure}

\section{Conclusion}\label{sec:Conclusion}
In this paper, we studied the problem of sparse support recovery in group testing under a novel measurement model, which characterizes the dilution effect, $\mathrm{LOD}$-induced binary quantization, and liquid dispensing volume errors. Compared with existing methods, our model provides a practical characterization of the laboratory testing process. Building on this model, we developed low-complexity decoding algorithms for both settings with known and unknown sparsity levels. For known sparsity level, we proposed a score-based decoding algorithm and established sufficient conditions on the number of tests required for exact recovery under both noiseless and noisy settings. For the more challenging scenario where the sparsity level $k$ is unknown, we developed a max-gap algorithm, along with a heuristic variant to enhance robustness, which can achieve exact support recovery with $\mathcal{O}\left(k\log n\right)$ measurements. Extensive simulations demonstrated that the proposed methods achieve reliable performance in the presence of the dilution, $\mathrm{LOD}$ and dispensing errors, significantly outperforming existing combinatorial group testing algorithms. Future work will investigate alternative pooling designs tailored to the proposed measurement model and further improve decoding performance under more general laboratory conditions.


{\appendices
\section{Proof of Lemma~\ref{lemma:probability gap}}\label{appendix:lemma_probability gap}
\begin{proof}
We prove Lemma~\ref{lemma:probability gap} by analyzing the conditional probabilities $\tilde{p}_{E_1}$ and $\tilde{p}_{E_2}$ under the noiseless and noisy settings.
\subsection{Noisy Case}
We first consider the noisy model. The noisy viral load in the $i$-th pool is given by  
\begin{equation}
  \tilde{z}_i=\frac{{\sum\limits_{j = 1}^n {\left( {1 + {\epsilon_{i,j}}} \right){A _{i,j}}{x_j}} }}{{\sum\limits_{j = 1}^n {\left( {1 + {\epsilon_{i,j}}} \right){A _{i,j}}} }}, \quad i=1,\ldots,m,
\end{equation}
where $\epsilon_{i,j}$ follows a Gaussian distribution $\mathcal{N}\left(0,\sigma_{\epsilon}^2\right)$ truncated to the interval $\left[-\epsilon_{max},\epsilon_{max}\right]$. The corresponding binary measurement vector $\tilde{y}$ is obtained as
\begin{equation}
  \tilde{y}_i=\mathbb{I}\left(\tilde{z_i}\geq \mathrm{LOD}\right),\quad i=1,\ldots,m.
\end{equation}

\subsubsection{Lower bound on $\tilde{p}_{E_1}$}
We define $\tilde{p}_{E_1}$ as the probability of $\tilde{y}_i=1$ given that the sample $j$ is positive and included in the $i$-th pool. For a positive sample $j\in\mathcal{S}$, the event $\tilde{y}_i=1$ occurs if
\begin{equation}
  \begin{aligned}
  &\left(1+\epsilon_{i,j}\right)x_j+\sum\limits_{l \ne j} {{\left(1+\epsilon_{i,l}\right)A _{i,l}x_l}}\\
  &\geq \mathrm{LOD}\cdot \left(\left(1+\epsilon_{i,j}\right) + {\sum\limits_{l \ne j} {{\left(1+\epsilon_{i,l}\right)A _{i,l}}}}\right).
\end{aligned} 
\end{equation}
A sufficient condition for this event is
\begin{equation}\label{eq:noisy_event}
  x_j\geq \text{LOD}\cdot \left(1+\sum\limits_{l \ne j} \frac{{\left( {1 + {\epsilon_{i,l}}} \right){A _{i,l}}}}{\left(1+\epsilon_{i,j}\right)}\right).
\end{equation}
Therefore, the probability $\tilde{p}_{E_1}$ is lower bounded by
\begin{equation}
  \tilde{p}_{E_1} \geq \mathbb{P}\left(x_j\geq \text{LOD}\cdot \left(1+\sum\limits_{l \ne j} \frac{{\left( {1 + {\epsilon_{i,l}}} \right){A _{i,l}}}}{\left(1+\epsilon_{i,j}\right)}\right)\right).
\end{equation}

Since $x_j$ is independent of $\left\{A_{i,l}\right\}_{l\ne j}$ and $\left\{\epsilon_{i,l}\right\}_{l\ne j}$, we introduce a threshold $d_T\geq 1$ and decompose the event in~(\ref{eq:noisy_event}) into two parts: 
\begin{itemize}
  \item[(1)] $x_j\geq \text{LOD}\cdot d_T$, 
  \item[(2)] $1+\sum\limits_{l \ne j} \frac{{\left( {1 + {\epsilon_{i,l}}} \right){A_{i,l}}}}{\left(1+\epsilon_{i,j}\right)} \leq d_T$,
\end{itemize}
which should hold simultaneously. Thus, we have
\begin{equation}\label{eq:two_event_pro}
  \begin{aligned}
    \tilde{p}_{E_1}\geq &\mathbb{P}\left(x_j\geq \text{LOD} \cdot d_T\right)\\ &\cdot \mathbb{P}\left(1+\sum\limits_{l \ne j} \frac{{\left( {1 + {\epsilon_{i,l}}} \right){A_{i,l}}}}{\left(1+\epsilon_{i,j}\right)} \leq d_T\right).
  \end{aligned}
\end{equation}

For the first term on the right-hand side of inequality~(\ref{eq:two_event_pro}), since $x_j$ follows a log-normal distribution, i.e., $\log_{10} x_j \sim \mathcal{N}\left(\mu_x, \sigma_x^2\right)$, then
\begin{equation}
  \begin{aligned}
  &\mathbb{P}\left(x_{j}\geq \text{LOD}\cdot d_T\right)=1-\Phi\left(\frac{\log_{10}\left(\text{LOD}\cdot d_T\right)-\mu_x}{\sigma_x}\right) \\
  &=\Phi\left(\frac{\mu_x-\log_{10}\left(\text{LOD}\cdot d_T\right)}{\sigma_x}\right),
  \end{aligned}
\end{equation}
where $\Phi\left(\cdot\right)$ is the cumulative distribution function of the standard Gaussian distribution. For simplicity, set $\Phi_{d_T}=\Phi\left(\left({\mu_x-\log_{10}\left(\text{LOD}\cdot d_T\right)}\right)/{\sigma_x}\right)$.

Applying Markov's inequality, the second term on the right-hand side of inequality~(\ref{eq:two_event_pro}) is lower bounded by
\begin{equation}
  \begin{aligned}
  &\mathbb{P}\left(1+\sum\limits_{l \ne j} \frac{{\left( {1 + {\epsilon_{i,l}}} \right){A_{i,l}}}}{\left(1+\epsilon_{i,j}\right)} \leq d_T \right) \\
  &\geq 1- \frac{1}{d_T} {\mathbb{E}\left[1+\sum\limits_{l \ne j} \frac{{\left( {1 + {\epsilon_{i,l}}} \right){A_{i,l}}}}{\left(1+\epsilon_{i,j}\right)} \right]} \\
  &=1- \frac{1}{{{d_T}}}\left( {1 + \left( {n - 1} \right)p\mathbb{E}\left[ {\frac{1}{\left( {1 + {\epsilon_{i,j}}} \right)}} \right]} \right).
  \end{aligned}
\end{equation}
Since $\left|\epsilon_{i,j}\right|<\epsilon_{max}<1$, we can expand $\frac{1}{1+\epsilon_{i,j}}$ as
\begin{equation}
  \frac{1}{{1 + {\epsilon_{i,j}}}} = \sum\limits_{t = 0}^\infty  {{{\left( { - {\epsilon_{i,j}}} \right)}^t}}.
\end{equation}
Taking expectation yields
\begin{equation}
  \mathbb{E}\left[\frac{1}{1+\epsilon_{i,j}}\right]=1+\sum\limits_{t = 1}^\infty  {\mathbb{E}\left[ {\epsilon_{i,j}^{2t}} \right]} \leq 1+\frac{\sigma_{\epsilon}^2}{1-\epsilon_{max}^2}.
\end{equation}

Combining the above, the probability $\tilde{p}_{E_1}$ satisfies
\begin{equation}
  \tilde{p}_{E_1}\geq \Phi_{d_T}\cdot \left(1-\frac{1+\left(n-1\right)p\left(1+\sigma_{\epsilon}^2/\left(1-\epsilon_{max}^2\right)\right)}{d_T}\right).
\end{equation}

\subsubsection{Upper bound on $\tilde{p}_{E_2}$}
We define $\tilde{p}_{E_2}$ as the probability of $\tilde{y}_i=1$ given that the sample $j$ is negative and included in the $i$-th pool. For a negative sample $j\notin \mathcal{S}$ with $x_j=0$, $\tilde{y}_i=1$ implies that at least one positive sample appears in the $i$-th pool. Hence,
\begin{equation}
  \tilde{p}_{E_2} \leq \mathbb{P}\left(\left|\{l:A_{i,l}=1\,and\,l\in \mathcal{S}\}\right|\geq 1\right)=1-\left(1-p\right)^k.
\end{equation} 

\subsubsection{Lower bound on $\Delta \tilde{p}$}
We now derive the lower bound on the probability gap $\Delta \tilde{p}$. By definition, $\Delta \tilde{p}=\tilde{p}_{E_1}-\tilde{p}_{E_2}$. Combining the bounds on $\tilde{p}_{E_1}$ and $\tilde{p}_{E_2}$ established above, we obtain
\begin{equation}\label{eq:probability_gap}
  \begin{aligned}
    &\Delta \tilde{p}=\tilde{p}_{E_1}-\tilde{p}_{E_2} \\
    &\geq {\Phi _{{d_T}}} \cdot \left( {1 - \frac{{1 + \left( {n - 1} \right)p}}{{{d_T}}} - \frac{{\left( {n - 1} \right)p\sigma _{\epsilon}^2}}{{\left( {1 - {\epsilon}_{\max }^2} \right){d_T}}}} \right) \\
    &\quad - \left(1 - \left(1-p\right)^k\right).
  \end{aligned}
\end{equation}

To ensure the reliable distinguishability between positive and negative samples, we require the probability gap $\Delta \tilde{p}$ to be sufficiently large. Specifically, we impose the condition
\begin{equation}
  \Delta\tilde{p}\geq \frac{\left(1-p\right)^k}{2},
\end{equation}
which holds when the probability $p$ satisfies
\begin{equation}
  0< p \leq \frac{2\left(d_T-1\right)\Phi_{d_T}-d_T}{2\left(n-1\right)\left(1+{\sigma_{\epsilon}^2}/\left({1-\epsilon_{max}^2}\right)\right)\Phi_{d_T}+kd_T}.
\end{equation}
To make certain that the above bound is nontrivial, we enforce
\begin{equation}
  \left(1-\frac{1}{d_T}\right)\Phi_{d_T}>\frac{1}{2}.
\end{equation}
Introducing a parameter $\alpha\in\left(\frac{1}{2},1\right)$, it suffices to require
\begin{equation}
  1-\frac{1}{d_T}>\alpha, \Phi_{d_T}>\frac{1}{2\alpha},
\end{equation}
which yields
\begin{equation}
  \frac{1}{1-\alpha}< d_T <\frac{1}{\mathrm{LOD}}10^{\mu_x-\sigma_x\Phi^{-1}\left({1}/{\left(2\alpha\right)}\right)},\alpha\in\left(\frac{1}{2},1\right).
\end{equation}

Under the above conditions on $p$ and $d_T$, we obtain
\begin{equation}
  \Delta\tilde{p}\geq \frac{\left(1-p\right)^k}{2}.
\end{equation}

\subsection{Noiseless Case}
The analysis for the noiseless setting is analogous to that of the noisy case, except that liquid dispensing volume errors are not considered. In the noiseless setting, the conditional probabilities $p_{E_1}$ and $p_{E_2}$ are bounded by
\begin{equation}
  \begin{aligned}
    p_{E_1}&\geq \Phi_{d_T}\cdot\left(1-\frac{1+\left(n-1\right)p}{d_T}\right),\\
    p_{E_2}&\leq 1-\left(1-p\right)^k,
  \end{aligned}
\end{equation}
respectively. These bounds imply that the probability gap $\Delta p$ satisfies
\begin{equation}
  \Delta p \geq \Phi_{d_T}\cdot\left(1-\frac{1+\left(n-1\right)p}{d_T}\right)-\left(1-\left(1-p\right)^k\right).
\end{equation}
Compared with the noisy case in~(\ref{eq:probability_gap}), the probability gap is increased by
\begin{equation*}
  \frac{\left(n-1\right)p\sigma_{\epsilon}^2}{\left(1-\epsilon_{max}^2\right)d_T}\Phi_{d_T}.
\end{equation*}
The proof of Lemma~\ref{lemma:probability gap} is thus completed. 
\end{proof}

\section{Proof of Theorem~\ref{theorem:score-based algorithm}}\label{appendix:theorem_score-based algorithm}
Before the proof of Theorem~\ref{theorem:score-based algorithm}, we introduce the Bernstein's inequality.
\begin{lemma}[Bernstein's inequality \cite{ci2013}]\label{lemma:bernstein inequality}
  Let $X_1,\ldots,X_n$ be independent real-valued random variables 
  with finite variance such that $X_i\leq b$ for some $b>0$ almost surely for all $i\leq n$. Let $S = \sum\limits_{i = 1}^n {\left( {{X_i} - \mathbb{E}\left[ {{X_i}} \right]} \right)}$ and $v = \sum\limits_{i = 1}^n {\mathbb{E}\left[ {X_i^2} \right]}$. Then, for all $t>0$, we have
  \begin{equation}
    \mathbb{P}\left( {S \geq t} \right) \leq \exp \left( { - \frac{{{t^2}}}{2\left({v + bt/3}\right)}} \right).
  \end{equation}
\end{lemma}

Similar bounds apply to the lower deviation $S \leq -t$ as well as the two-sided deviation $\left|S\right|\geq t$, with an additional factor of two in the latter case.

\begin{proof}
  For each sample $j\in\left\{1,2,\ldots,n\right\}$, we define $X_j$ as its score, where the contribution from the $i$-th test is given by
  \begin{equation}
    X_j^{(i)}=A_{i,j}\left(a\tilde{y}_i-b\left(1-\tilde{y}_i\right)\right), \quad i=1,\ldots,m.
  \end{equation}
  If the sample $j$ is positive, i.e., $j\in \mathcal{S}$, we obtain
  \begin{equation}
    \begin{aligned}
      &\mu_1\triangleq \mathbb{E}\left[X_j^{(i)}\mid j\in \mathcal{S}\right]=\mathbb{E}\left[A_{i,j}\left(a\tilde{y}_i-b(1-\tilde{y}_i)\right)\mid j\in \mathcal{S}\right]\\
      &=p\left((a+b)\tilde{p}_{E_1}-b\right),\\
      &v_1\triangleq\mathbb{E}\left[\left(X_j^{(i)}\right)^2 \mid j\in \mathcal{S}\right]=p\left(a^2\tilde{p}_{E_1}+ b^2\left(1-\tilde{p}_{E_1}\right)\right).
    \end{aligned}
  \end{equation}
  Similarly, for $j\notin \mathcal{S}$, 
  \begin{equation}
    \begin{aligned}
      &\mu_2\triangleq\mathbb{E}\left[X_j^{(i)}\mid j\notin \mathcal{S}\right]=p\left((a+b)\tilde{p}_{E_2}-b\right),\\
      &v_2\triangleq\mathbb{E}\left[\left(X_j^{(i)}\right)^2 \mid j\notin \mathcal{S}\right]=p\left(a^2\tilde{p}_{E_2}+ b^2\left(1-\tilde{p}_{E_2}\right)\right).
    \end{aligned}
  \end{equation}

  To exactly recover the support set $\mathcal{S}$, it suffices that each positive sample attains a higher score than any negative sample, i.e.,
  \begin{equation}
    \mathop {\min }\limits_{j \in \mathcal{S}} {X_j} > \mathop {\max }\limits_{j \notin \mathcal{S}} {X_j},
  \end{equation}
  where ${X_j} = \sum\limits_{i = 1}^m {X_j^{(i)}}$. Accordingly, recovery fails if at least one negative sample $j\notin\mathcal{S}$ has $X_j\geq \mathop {\min }\limits_{j \in \mathcal{S}} {X_j}$. The probability of this event satisfies
  \begin{equation}\label{eq:score pro}
    \begin{aligned}
      &\mathbb{P}\left(\mathop {\min }\limits_{j \in \mathcal{S}} {X_j} \leq \mathop {\max }\limits_{j \notin \mathcal{S}} {X_j}\right) \\
      &\leq \mathbb{P}\left(\mathop {\min }\limits_{j \in \mathcal{S}} {X_j} \leq X_T\right) + \mathbb{P}\left(\mathop {\max }\limits_{j \notin \mathcal{S}} {X_j} \geq X_T\right),
    \end{aligned}
  \end{equation} 
  for any threshold $X_T>0$. We choose $X_T=\frac{\left(\mu_1+\mu_2\right)m}{2}$. Applying Lemma~\ref{lemma:bernstein inequality}, the first term on the right-hand side of inequality~(\ref{eq:score pro}) is bound by
  \begin{equation}
    \begin{aligned}
      &\mathbb{P}\left(\mathop {\min }\limits_{j \in \mathcal{S}} {X_j} \leq X_T\right) \leq k \mathbb{P}\left(X_j \leq X_T \mid j\in \mathcal{S}\right) \\
      &\leq k\exp\left(-\frac{t^2 / 2}{mv_1+tM/3}\right),
    \end{aligned}
  \end{equation}
  where $t=\frac{\left(\mu_1-\mu_2\right)m}{2}=\frac{mp\left(a+b\right){\Delta \tilde{p}}}{2}$ and $M=\max \left(a,b\right)$. Similarly, for the second term,
  \begin{equation}
    \begin{aligned}
      &\mathbb{P}\left(\mathop {\max }\limits_{j \notin \mathcal{S}} {X_j} \geq X_T\right) \leq \left(n-k\right) \mathbb{P}\left(X_j\geq X_T \mid j\notin \mathcal{S}\right) \\
      &\leq \left(n-k\right)\exp\left(-\frac{t^2 / 2}{mv_2+tM/3}\right).
    \end{aligned}
  \end{equation}
  Thus, the probability of at least one false negative is bounded by
  \begin{equation}
    \begin{aligned}
      &\mathbb{P}\left(\mathrm{there\,exists\,at\,least\,one\,false\,negative}\right) \\
      &\leq k\exp\left(-\frac{t^2 / 2}{mv_1+tM/3}\right) + \left(n-k\right)\exp\left(-\frac{t^2 / 2}{mv_2+tM/3}\right)\\
      &\leq n\exp\left(-\frac{t^2 / 2}{m \mathop{\max}\left(v_1,v_2\right)+tM/3}\right).
    \end{aligned}
  \end{equation}
  
  Let $\delta\in\left(0,1\right)$ denote the target failure probability. We have $\mathbb{P}\left(\hat{\mathcal{S}}=\mathcal{S}\right)\geq 1-\delta$, provided that
  \begin{equation}
    m \geq \frac{4 \left(2r^2 + {\Delta \tilde{p}}r/3\right) \log\left(n/\delta\right)}{p \left({\Delta \tilde{p}}\right)^2},
  \end{equation}
  where $r=\frac{\max\left(a, b\right)}{a+b}$. This completes the proof of Theorem~\ref{theorem:score-based algorithm}.
\end{proof}

\section{Proof of Theorem~\ref{theorem:max-gap algorithm}}\label{appendix:theorem_max-gap algorithm}
\begin{proof}
  For the max-gap algorithm, exact support recovery is achieved if the scores of all samples concentrate around their expectations and there exists a clear separation between positive and negative samples. Specifically, for $j\in \mathcal{S}$, $\mathbb{E}\left[X_j\right]=m\mu_1$, and for $j\notin \mathcal{S}$, $\mathbb{E}\left[X_j\right]=m\mu_2$. It suffices to establish that, with high probability, 
  \begin{equation}\label{eq:gap pro}
    \begin{aligned}
      &X_j\in\left(m\mu_1-t, m\mu_1+t\right), j\in \mathcal{S}, \\
      &X_j\in\left(m\mu_2-t, m\mu_2+t\right), j\notin \mathcal{S},
    \end{aligned}
  \end{equation}
  for some $t>0$. To guarantee a clear separation between two groups, we impose $m\mu_1-m\mu_2\geq 4t$, which yields
  \begin{equation}
    t\leq \frac{mp\left(a+b\right)\Delta \tilde{p}}{4}.
  \end{equation}

  Applying Lemma~\ref{lemma:bernstein inequality}, we have
  \begin{equation}
    \begin{aligned}
      &\mathbb{P}\left(\left|X_j-m\mu_1\right| \geq t \mid j\in \mathcal{S}\right)\leq 2\exp\left(-\frac{t^2 / 2}{mv_1+tM/3}\right),\\
      &\mathbb{P}\left(\left|X_j-m\mu_2\right| \geq t \mid j\notin \mathcal{S}\right)\leq 2\exp\left(-\frac{t^2 /2}{mv_2+tM/3}\right).
    \end{aligned}
  \end{equation}
  By the union bound over all samples, the probability that the events in~(\ref{eq:gap pro}) fail is bounded by 
  \begin{equation}
    \begin{aligned}
      \mathbb{P}\left(\mathrm{failure}\right) &\leq 2k\exp\left(-\frac{t^2 / 2}{mv_1+tM/3}\right)\\ 
      &\quad + 2\left(n-k\right)\exp\left(-\frac{t^2 /2}{mv_2+tM/3}\right)\\
      &\leq 2n\exp\left(-\frac{t^2 / 2}{m\max\left(v_1, v_2\right)+tM/3}\right).
    \end{aligned}
  \end{equation}

  Let $\delta\in\left(0,1\right)$ denote the target failure probability. Substituting $t = \frac{m p (a+b)\Delta \tilde{p}}{4}$, we can obtain $\mathbb{P}\left(\hat{\mathcal{S}}=\mathcal{S}\right)\geq 1-\delta$, provided that
  \begin{equation}
    m \geq \frac{16 \left(2r^2 + {\Delta \tilde{p}}r/6\right) \log\left(2n/\delta\right)}{p \left({\Delta \tilde{p}}\right)^2},
  \end{equation}
  where $r=\frac{\max\left(a, b\right)}{a+b}$. This completes the proof of Theorem~\ref{theorem:max-gap algorithm}.
\end{proof}}



\bibliographystyle{IEEEtran}
\bibliography{IEEEabrv,reference}


 





\end{document}